\def\icdcn{0} 
\def\icdcnintro{1}
\def\nopfs{0} 
\def\esa{0} 
\def\pagenr{0} 

\ifnum\icdcn=0
\documentclass[12pt]{article}
\usepackage{times}
\usepackage[a4paper,hmargin=1in,vmargin=1.1in]{geometry}
\usepackage{amsthm}
\else
\documentclass{sig-alternate}
\fi

\ifnum\pagenr=1
\usepackage{fancyhdr}
\fi

\usepackage{multirow}
\usepackage{enumerate}
\usepackage{comment}
\usepackage{marginnote}
\usepackage{url}
\usepackage{subfig}
\usepackage{amsfonts}
\usepackage{amsmath}
\usepackage{amssymb}
\usepackage{wrapfig}

\usepackage{graphicx}
\graphicspath{{./Figs/}}
\usepackage{paralist}
\usepackage[active]{srcltx}
\usepackage{color}
\usepackage{algorithm}
\usepackage{algpseudocode}
\usepackage{setspace}

\algrenewcommand\algorithmicrequire{\textbf{\quad Input:}}
\algrenewcommand\algorithmicensure{\textbf{\quad Output:}}
\newenvironment{proof sketch}[1]{\noindent {\emph{Proof sketch of #1:}}}{\hfill \qed}

\ifnum\esa=0
\newtheorem{theorem}{Theorem}
\newtheorem{proposition}[theorem]{Proposition}
\newtheorem{lemma}[theorem]{Lemma}

\newtheorem{definition}{Definition}

\fi

\newcommand{\eqdf}{\triangleq}
\newcommand{\eps}{\varepsilon}

\newcommand{\poly}{{\rm poly}}
\newcommand{\polylog}{{\rm polylog}}

\newcommand{\sol}{\textit{sol}}

\newcommand{\prev}{\textit{prev}}

\newcommand{\alg}{\textsc{alg}}

\newcommand{\degree}{\text{\textit{deg}}}
\newcommand{\dist}{\text{\textit{dist}}}
\newcommand{\wmin}{w_{\min}}

\newcommand{\dlocal}{\textsc{DistLocal}}
\newcommand{\dlocall}{\text{DistLocal}}
\newcommand{\clocal}{\textsc{CentLocal}}
\newcommand{\clocall}{\text{CentLocal}}

\newcommand{\orad}{\text{\textsc{o-rad}}}
\newcommand{\oradd}{\text{O-RAD}}
\newcommand{\ao}{\text{\textsc{ao}}}
\newcommand{\rad}{\text{\textit{rad}}}

\newcommand{\mis}{\text{\textsc{mis}}}
\newcommand{\miss}{\text{\text{MIS}}}
\newcommand{\lmis}{\text{\textsc{l-mis}}}
\newcommand{\mm}{\text{\textsc{mm}}}
\newcommand{\mcm}{\text{\textsc{mcm}}}
\newcommand{\mcmm}{\text{\text{MCM}}}
\newcommand{\mum}{\text{\textsc{mcm}}}
\newcommand{\mwm}{\text{\textsc{mwm}}}
\newcommand{\mwmm}{\text{\text{MWM}}}
\newcommand{\pcolor}{\text{\poly$(\Delta)$-\textsc{Color}}}
\newcommand{\dcolor}{\text{$(\Delta+1)$\textsc{-Color}}}

\newcommand{\NN}{{\mathbb{N}}}

\newcommand{\oracle}{{\mathcal O}}
\newcommand{\proca}{{\mathcal A}}

\ifnum\icdcn=0
\def\mnFONT{\footnotesize}
\newcommand{\dnote}[1]{\marginpar{\mnFONT\begin{minipage}[t]{0.75in}
                      \raggedright\hrule
                       {\sf D:} {\sl{#1}}\end{minipage}}}
\newcommand{\mnote}[1]{\marginpar{\mnFONT\begin{minipage}[t]{0.75in}
                      \raggedright\hrule
                       {\sf M:} {\sl{#1}}\end{minipage}}}
\newcommand{\gnote}[1]{\marginpar{\mnFONT\begin{minipage}[t]{0.75in}
                      \raggedright\hrule
                       {\sf G:} {\sl{#1}}\end{minipage}}}
\else
\newcommand{\dnote}[1]{}
\newcommand{\mnote}[1]{}
\newcommand{\gnote}[1]{}
\fi

\ifnum\nopfs=1
\usepackage{environ}
\NewEnviron{killcontents}{}

\fi

\ifnum\icdcn=1
\conferenceinfo{ICDCN}{'15 Jan. 4-7, Goa, India}
\CopyrightYear{2015}
\crdata{978-1-4503-2928-6}
\fi

\begin{document}
\title{
Distributed Maximum Matching in
Bounded Degree Graphs
}

\ifnum\icdcn=0
\author{ %
Guy Even\thanks{School of Electrical Engineering, Tel-Aviv
Univ., Tel-Aviv 69978, Israel.
\protect\url{{guy,medinamo,danar}@eng.tau.ac.il}.}
\and Moti Medina$^*$\thanks{M.M was partially funded
by the Israeli Ministry of Science and Technology.}
\and Dana Ron$^*$\thanks{Research supported by the Israel Science Foundation grant number 671/13.}
}
\else 
\numberofauthors{3}
\author{
\alignauthor Guy Even\\
       \affaddr{School of Electrical Engineering}\\
       \affaddr{Tel-Aviv University}\\
       \affaddr{Tel-Aviv 69978, Israel}\\
       \email{guy@eng.tau.ac.il}
\alignauthor Moti Medina\titlenote{M.M was partially funded
by the Israeli Ministry of Science and Technology.}\\
       \affaddr{School of Electrical Engineering}\\
       \affaddr{Tel-Aviv University}\\
       \affaddr{Tel-Aviv 69978, Israel}\\
       \email{medinamo@eng.tau.ac.il}
\alignauthor Dana Ron \titlenote{Research supported by the Israel Science Foundation grant number 671/13.}\\
       \affaddr{School of Electrical Engineering}\\
       \affaddr{Tel-Aviv University}\\
       \affaddr{Tel-Aviv 69978, Israel}\\
       \email{danar@eng.tau.ac.il}
}

\fi

\maketitle
\ifnum\pagenr=1
\thispagestyle{fancy}
\fi
\begin{abstract}
We present deterministic distributed algorithms for computing approximate maximum cardinality
matchings and approximate maximum weight matchings. Our algorithm for the unweighted case
computes a matching whose size is at least $(1-\eps)$ times the optimal in
$\Delta^{O(1/\eps)} + O\left(\frac{1}{\eps^2}\right) \cdot\log^*(n)$ rounds
where $n$ is the number of vertices in the graph and $\Delta$ is the maximum degree.
Our algorithm for the edge-weighted case
computes a matching whose weight is at least $(1-\eps)$ times the optimal in
$\log(\min\{1/\wmin,n/\eps\})^{O(1/\eps)}\cdot(\Delta^{O(1/\eps)}+\log^*(n))$ rounds
for edge-weights in $[\wmin,1]$.

The best previous algorithms for both the unweighted case
and the weighted case are by Lotker, Patt-Shamir, and
Pettie~(SPAA 2008).  For the unweighted case they give a
randomized $(1-\eps)$-approximation algorithm that runs
in $O((\log(n)) /\eps^3)$ rounds. For the weighted case
they give a randomized  $(1/2-\eps)$-approximation
algorithm that runs in $O(\log(\eps^{-1}) \cdot \log(n))$
rounds. Hence, our results improve on the previous ones
when the parameters $\Delta$, $\eps$ and $\wmin$ are
constants (where we reduce the number of runs from
$O(\log(n))$ to  $O(\log^*(n))$), and more generally when
$\Delta$, $1/\eps$ and $1/\wmin$ are sufficiently slowly
increasing functions of $n$. Moreover, our algorithms are
deterministic rather than randomized.
\end{abstract}

\ifnum\icdcn=1

 \fi

\ifnum\icdcn=0
 \else  \fi

\section{Introduction}\sloppy
\ifnum\icdcnintro=1 In this work we consider the problem of distributively computing
an approximate maximum (weighted) matching in bounded degree graphs.  Let $G=(V,E)$
denote an edge weighted graph with $n$ vertices and maximum degree $\Delta$.  Assume
that the maximum edge weight is $1$ and let $\wmin$ denote the minimum edge weight.
Denoting by $\mcm(G)$ the maximum cardinality of a matching in $G$ and by $\mwm(G)$
the maximum weight of a matching in $G$, we present the following results:
\begin{itemize}
\item A deterministic distributed algorithm that for any $\eps \in (0,1)$ computes a
  matching whose size is at least $(1-\eps)\cdot \mcm(G)$ in $$\Delta^{O(1/\eps)} +
  O\left(\frac{1}{\eps^2}\right) \cdot\log^*(n)$$ rounds.
\item \sloppy A deterministic distributed algorithm that and for any $\eps \in (0,1)$
   computes a matching whose weight is at least $(1-\eps)\cdot \mwm(G)$ in
   $$\log(\min\{1/\wmin,n/\eps\})^{O(1/\eps)}\cdot(\Delta^{O(1/\eps)}+\log^*(n))$$
  rounds.
\end{itemize}
The best previous algorithms for both the unweighted and
weighted cases are by Lotker, Patt-Shamir, and
Pettie~\cite{DBLP:conf/spaa/LotkerPP08}. For the unweighted
case they give a randomized $(1-\eps)$-approximation
algorithm that runs in $O((\log(n))/\eps^{3})$ rounds with
high probability%
\footnote{We say that an event occurs with high probability
if it occurs with probability at least
$1-\frac{1}{\poly(n)}$. }
 (w.h.p). Hence we get an improved result
when $\Delta^{O(1/\eps)} = o(\log(n))$. In particular, for
constant $\Delta$ and $\eps$, the number of rounds is
$O(\log^*(n))$.  Note that an $O(1)$-approximation of a
maximum matching in an $n$-node ring cannot be computed by
any
deterministic
distributed algorithm in $o(\log^*(n))$
rounds~\cite{czygrinow2008fast,lenzen2008leveraging}.  For
the weighted case, they give a randomized
$(1/2-\epsilon)$-approximation algorithm that runs in
$O(\log(\eps^{-1}) \cdot \log(n))$ rounds
(w.h.p)\footnote{Lotker, Patt-Shamir and Pettie
remark~\cite[Sec. 4]{DBLP:conf/spaa/LotkerPP08} that a
$(1-\eps)$-MWM can be obtained in $O(\eps^{-4}\log^2 n)$
rounds (using messages of linear size), by adapting the
algorithm of Hougardy and Vinkemeir~\cite{HV06} (where
details are not provided in the paper).}. Our
 \mwm\ approximation algorithm 
runs in significantly fewer rounds
for various settings of the parameters $\Delta$, $1/\eps$, and $1/\wmin$.
In particular, when they are constants, the number of rounds is
$O(\log^*(n))$.

\renewcommand{\arraystretch}{2}
\begin{table*}[htb]
\ifnum\icdcn=0
\footnotesize
\fi
\centering
\begin{tabular}{| c | c | c | c || c| c | c |c |}
\hline
\multicolumn{4}{|c ||}{Previous work } &  \multicolumn{2}{|c |}{Here (Deterministic)}\\
\cline{1-6}
problem & \# rounds & success prob. & apx. ratio. & \# rounds & apx. ratio.\\
\hline\hline
\multirow{2}{*}{\mcm} &$O(\frac{\log(n)}{\eps^{3}})$ & $1-\frac{1}{\poly(n)}$ & $1-\eps$ ~\cite{DBLP:conf/spaa/LotkerPP08}&  \multirow{2}{*}{$\Delta^{O\left(\frac{1}{\eps}\right)}+O\left(\frac{1}{\eps^2}\right)\cdot \log^*(n)$} & \multirow{2}{*}{$1-\eps$} \\
& $\Delta^{O(\frac{1}{\eps})}$ & $1-\Theta(1)$ & $1-\eps$ ~\cite{onak2008}& &   [Thm.~\ref{thm:distalg}]\\
%

\hline
\multirow{2}{*}{\mwm} &$O\left(\log(\eps^{-1}) \cdot \log(n)\right)$ & $1-\frac{1}{\poly(n)}$ & $1/2-\eps$ ~\cite{DBLP:conf/spaa/LotkerPP08}&  \multirow{2}{*}{$\log^{O\left(\frac{1}{\eps}\right)}(\Gamma)\cdot\left(\Delta^{O\left(\frac{1}{\eps}\right)}+\log^*(n)\right)$} & \multirow{2}{*}{$1-\eps$} \\
%
%
%
& $O\left( \frac{\log^4 (n)}{\eps}\cdot \log
(\Gamma)\right)$ & deterministic & $1/6-\eps$ ~\cite{panconesi2010fast}& &  [Thm.~\ref{thm:distalg mwm}]\\
%
\hline
\end{tabular}
\caption{\small A comparison between \mum\ and \mwm\
\dlocal\ algorithms. The ratio between the maximum to
minimum edge weight is denoted by $\Gamma$
(we may assume that $\Gamma\leq n/\eps$).
}
   \label{tab:dist}\label{tbl}
\end{table*}
\renewcommand{\arraystretch}{1}

\subsection{Techniques}
\subsubsection{Centralized local computation algorithms}
For both the unweighted and the weighted versions of the problem we design
(deterministic) {\em Centralized Local Computation
Algorithms\/}, which translate into distributed algorithms.
Centralized local computation (\clocal) algorithms,  as defined by Rubinfeld et
al.~\cite{shaiics2011}, are algorithms that answer queries regarding (global)
solutions to computational problems by performing local (sublinear time) computations
on the input.
The answers to all queries must be consistent with a single solution regardless of the number of possible solutions.
In particular, for the problems we study, each query is an edge $e$ in the graph $G$, and
the  algorithm needs to answer whether $e$ belongs to a matching $M$ whose size (or weight)
is at least $(1-\eps)$ times the optimal. Consistency means that all answers to the queries must be according
to the same matching $M$. To this end the algorithm can
probe the graph $G$ by asking about the neighbors of vertices of its choice.
In this manner the algorithm can obtain the local neighborhood of each queried edge $e$.

\sloppy A \clocal\ algorithm may be randomized, so that the solution according to which it
answers queries may depend on its internal coin flips. However, the solution should
not depend on the sequence of the queries (this property is called query order
obliviousness~\cite{shaiics2011}).
The main performance measure of \clocal\ algorithms is the maximum number of probes performed
per query. In this work we will actually be interested in the probe-radius, that is, the maximum
distance in the graph of a probe from the queried edge. This translates into the number of rounds performed
by the corresponding distributed algorithm.

%

We believe that using the design methodology of first describing and analyzing a \clocal\
algorithm and then transforming it into a distributed local\footnote{Strictly speaking, a distributed algorithm is considered
 {\em local\/} if it performs a number of rounds that does not depend on $n$.
 Here we allow a weak dependence on $n$ (i.e., $\log^*(n)$ or even $\polylog(n)$).}
  (\dlocal) algorithm makes
both the presentation and the analysis simpler and easier to follow.
The benefit of designing \clocal\ algorithms is that it removes the need for coordination between vertices when performing
computations on auxiliary graphs (for discussion on these graphs see the following subsection).
The transformation from a \clocal\ algorithm to a \dlocal\ one is especially straightforward
when the \clocal\ algorithm is deterministic (and stateless -- Section~\ref{sec:simu}).


\subsubsection{A Global Algorithm for Approximate Maximum Cardinality Matching}
Previous \clocal\ and \dlocal-algorithms for finding an approximate maximum cardinality
matching are based on the following
framework~\cite{onak2008,DBLP:conf/spaa/LotkerPP08,shaiapprox2013}.
First consider a global/abstract algorithm
whose correctness is based on a result of Hopcroft and Karp~\cite{hopcroft1973n}.
The algorithm works iteratively, where in each iteration it constructs a new matching (starting from the empty matching).
Each new matching is constructed based on a maximal set of vertex disjoint paths that are augmenting paths with respect to the previous matching. Such a maximal set is a maximal independent set (\mis) in the intersection graph over
the augmenting paths.  (See Algorithm.~\ref{alg:global mcm} 
for precise details.)
The question is how to simulate this global algorithm in a local/distributed fashion, and
in particular, how to compute the maximal independent sets  over the intersection graphs.

\subsubsection{Local Simulation}
Our \clocal\ algorithm for approximate \mcm\
follows Nguyen and Onak's~\cite{onak2008}
sublinear algorithm for approximating the
size of a maximum matching (see also~\cite{shaiapprox2013}).
The algorithm works recursively, where
recursion is applied to determine membership in the previous matching (defined by
the global algorithm) as well as
membership in an augmenting path that belongs to the maximal set of augmenting paths.

We differ from~\cite{onak2008} and~\cite{shaiapprox2013} in the \clocal\ \mis\ algorithm that we apply,
which is the algorithm presented
in~\cite{EMR14}.
Recall that the \mis\ algorithm is applied to intersection graphs over
augmenting paths. The \mis\ algorithm works by computing an acyclic orientation of
the edges of the graph, where the radius of the orientation (the longest directed
path in the oriented graph) is $\poly(\Delta)$. This in turn is performed by coloring
the vertices in $\poly(\Delta)$ colors.  An acyclic orientation induces a partial
ordering over the vertices, which enables to (locally) simulate the simple greedy
sequential algorithm for \mis.  The main issue is analyzing the total probe-radius of
the resulting combined \clocal\ algorithm.

\subsubsection{Weighted Matchings}
Our \clocal\ algorithm for approximate \mwm\ is also based on the abovementioned \mis\ algorithm as
an ``inner'' building block. As the ``outer'' building block we use a result
described in~\cite{onakthesis} (and mentioned in~\cite{onak2008}) for approximating
the maximum weight of a matching, which in turn builds on work of Pettie and
Sanders~\cite{DBLP:journals/ipl/PettieS04}.

\subsection{Related Work}
We compare our results to previous ones in Table~\ref{tbl}.
The first line refers to the aforementioned algorithm by Lotker, Patt-Shamir, and
Pettie~\cite{DBLP:conf/spaa/LotkerPP08} for the unweighted case.
The second line in Table~\ref{tbl}
refers to an algorithm of Nguyen and Onak~\cite{onak2008}.
As they observe, their algorithm for approximating the size
of a maximum matching in sublinear time can be transformed
into a randomized distributed algorithm that succeeds with constant
probability (say, $2/3$) and runs in $\Delta^{O(1/\eps)}$
rounds.
The third line refers to the aforementioned algorithm by Lotker, Patt-Shamir, and
Pettie~\cite{DBLP:conf/spaa/LotkerPP08} for the weighted case.
The fourth line refers to the algorithm by Panconesi and
Sozio~\cite{panconesi2010fast} for weighted matching. They devise a deterministic
distributed $(1/6-\eps)$-approximation algorithm
that runs in $O\left(
  \frac{\log^4 (n)}{\eps}\cdot \log (\Gamma)\right)$ rounds, where $\Gamma$ is the ratio between the
maximum to minimum edge weight.

We remark that the randomized \clocal-algorithm by Mansour
and Vardi~\cite{shaiapprox2013} for $(1-\eps)$-approximate
maximum cardinality matching in bounded-degree graphs can be transformed into a
randomized \dlocal-algorithm for $(1-\eps)$-approximate
maximum cardinality matching (whose success probability is $1-1/\poly(n)$).
Their focus is on bounding the number
of probes, which they show is polylogarithmic in $n$ for constant $\Delta$ and $\eps$.
To the best of our understanding, an analysis of the probe-radius of their
algorithm will not imply a \dlocal-algorithm that runs in fewer rounds
than the algorithm of
Lotker, Patt-Shamir, and Pettie~\cite{DBLP:conf/spaa/LotkerPP08}.

\else 
{\em Local Computation Algorithms\/}, as defined by Rubinfeld et
al.~\cite{shaiics2011}, are algorithms that answer queries regarding (global)
solutions to computational problems by performing local (sublinear time) computations
on the input. The answers to all queries must be consistent with a single solution regardless of the number of possible solutions.  To make this notion
concrete, consider the {\em Maximal Independent Set\/} problem, which we denote by
\mis.  Given a graph $G = (V,E)$ as input, the local algorithm \alg\ gives the
illusion that it ``holds'' a specific maximal independent set $I \subseteq V$.  Namely, given
any vertex $v$ as a query, \alg\ answers whether $v$ belongs to $I$ even though
\alg\ cannot read all of $G$, cannot store the solution $I$, and cannot even remember all the answers to previous queries.  In order to answer such queries, $\alg$
can probe the graph $G$ by asking about the neighbors of a vertex of its choice.

A local algorithm may be randomized, so that the solution according to which it
answers queries may depend on its internal coin flips. However, the solution should
not depend on the sequence of the queries (i.e., this property is called query order oblivious~\cite{shaiics2011}).  The performance of a local
computation algorithm is measured by the following criteria: the maximum number of probes it makes
to the input per query, the success probability over any sequence of queries, and the
maximum space it uses between queries. It is desired that both the probe complexity
and the space complexity of the algorithm be sublinear in the size of the graph
(see e.g., ${\rm polylog}(|V|)$), and that the success probability be $1-1/\poly(|V|)$.
It is usually assumed that the maximum degree of the graph is upper-bounded by a
constant, but our results are useful also for non-constant upper bounds.
For a formal definition of local algorithms in the context of graph problems,
which is the focus of this work, see Subsection~\ref{sec:cmodel}.

The motivation for designing local computation algorithms is that local computation algorithms capture difficulties with very large inputs.
A few examples include:
\begin{inparaenum}[(1)]
  \item Reading the whole input is too costly if the input is very long.
  \item In certain situations one is interested in a very small part of a complete solution.
  \item Consider a setting in which different uncoordinated servers need to answer queries about a very long input stored in the cloud.
      The servers do not communicate with each other, do not store answers to previous queries, and want to minimize their accesses to the input.
\end{inparaenum}

Local computation algorithms have been designed for various graph (and hypergraph) problems,
including the abovementioned \mis~\cite{shaiics2011,shaisoda2012},
hypergraph coloring~\cite{shaiics2011,shaisoda2012},
maximal matching~\cite{shaiicalp2012}
and (approximate) maximum matching~\cite{shaiapprox2013}.
Local computation algorithms also appear  implicitly in works on sublinear
approximation algorithms for various graph parameters, such as the size of
a minimum vertex cover~\cite{parnasron,onak2008,yoshida2009improved,dana2012}.
Some of these implicit results are very efficient in terms of their probe complexity
(in particular, it depends on the maximum degree and not on $|V|$)
but do not give the desired $1-1/\poly(|V|)$ success probability. We compare our
results to both the explicit and implicit relevant known results.

\ifnum\esa=0
As can be gleaned from the definition in~\cite{shaiics2011},
local computation algorithms are closely related~\cite{parnasron} to {\em Local Distributed Algorithms\/}.
We discuss the similarities and differences in more detail in Subsection~\ref{subsec:rel-clocal-dlocal}.
In this work, we exploit this relation in two ways. First, we
use techniques from the study of local distributed algorithms to obtain better local computation algorithms.
Second, we apply techniques
from the study of local computation algorithms (more precisely, local computation algorithms that
are implicit within sublinear approximation algorithms) to
obtain a new result in distributed computing.
\else
As can be gleaned from the definition in~\cite{shaiics2011},
local computation algorithms are closely related to {\em Local Distributed Algorithms}~\cite{parnasron}.
Such a connection is discussed in Section~\ref{sec:dlocal-def}.
\fi

\sloppy In what follows we denote the aforementioned local
computation model by \clocal\ (where the ``\textsc{Cent}''
stands for ``centralized'') and the distributed (local)
model
 by \dlocal\ (for a formal definition of the latter,
see Subsection~\ref{sec:dlocal-def}).
We denote the number of vertices in the input graph by $n$ and the maximum degree by
$\Delta$.

\ifnum\esa=0 \sloppy
\subsection{On the relation between \clocall\ and \dlocall}\label{subsec:rel-clocal-dlocal}
 The \clocal\ model is centralized in the sense
that there is a single central algorithm that is provided access to the whole graph.
This is as opposed to the \dlocal\ model in which each processor resides in
a graph vertex $v$ and can obtain information only about the neighborhood of $v$.
Another important difference is in the main complexity measure. In the \clocal\ model,
one counts the number of probes that the algorithm performs per query, while
in the \dlocal\ model, the number of rounds of communication is counted.
This implies that a \dlocal\ algorithm always obtains information about a ball centered
at a vertex, where the radius of the ball is the number of rounds of communication. On the other hand,
in the case of a \clocal\ algorithm, it might choose to obtain information about different
types of neighborhoods so as to save in the number of probes. Indeed (similarly
to what was observed in the context of sublinear approximation algorithms~\cite{parnasron}),
given a \dlocal\ algorithm
for a particular problem with round complexity $r$, we directly obtain a \clocal\ algorithm
whose probe complexity is $O(\Delta^r)$ where $\Delta$ is the maximum degree in the graph.
However, we might be able to obtain lower probe complexity if we don't apply such a
black-box reduction. In the other direction, \clocal\ algorithms with certain properties, can be transformed into
\dlocal\ algorithms.
\fi

\subsection{The Ranking Technique}
The starting point for our results in the \clocal\ model is the
\emph{ranking} technique~\cite{onak2008,yoshida2009improved,shaisoda2012,shaiicalp2012,shaiapprox2013}. To exemplify this, consider, once again, the \mis\ problem.
A very simple (global ``greedy'') algorithm for this problem works by
selecting an arbitrary ranking of the vertices and initializing $I$ to be empty.
The algorithm then considers the vertices one after the other according to their
ranks and adds a vertex to $I$ if and only if it does not neighbor any vertex
already in $I$. Such an algorithm can be ``localized'' as follows. For a fixed
ranking of the vertices (say, according to their IDs), given a query on a vertex $v$,
the local algorithm performs a {\em restricted\/} DFS starting from $v$. The
restriction is that the search continues only on paths with monotonically decreasing
ranks. The local algorithm then simulates the global one on the subgraph induced by
this restricted DFS.

The main problem with the above local algorithm is that the number of probes it performs
when running the DFS may be very large. Indeed, for some rankings (and queried vertices), the number
of probes is linear in $n$. In order to circumvent this problem,
{\em random\/} rankings were studied~\cite{onak2008}. This brings up two questions, which were studied in previous works,
both for the \mis\ algorithm described above~\cite{onak2008,yoshida2009improved}
and for other ranking-based
algorithms~\cite{shaisoda2012,shaiicalp2012,shaiapprox2013}.
The first is to bound
the number of probes needed to answer a query with high probability. The second is
how to efficiently store a random ranking between queries.

\subsection{Our Contributions}
\label{subsec:results-clocal}
\paragraph{Deterministic Stateless \clocal\ Algorithms as a
Design Methodology for \dlocal\ Algorithms}

\paragraph{Deterministic \dlocal\ Approximation Scheme for \mcm}

\paragraph{Deterministic \dlocal\ Approximation Scheme for \mwm}

\paragraph{Orientations with bounded reachability}
Our first conceptual contribution is a simple but very useful observation. Rather
than considering vertex rankings, we suggest to consider {\em acyclic orientations\/}
of the edges in the graph. Such orientations induce partial orders over the vertices,
and partial orders suffice for our purposes. The probe complexity induced by a given
orientation translates into a combinatorial measure, which we refer to as the {\em
  reachability\/} of the orientation. Reachability of an acyclic orientation is the
maximum number of vertices that can be reached from any start vertex by
directed paths (induced by the orientation). This leads us to the quest for a
\clocal\ algorithm that computes an orientation with bounded reachability.

\paragraph{Orientations and colorings}
Our second conceptual contribution is that an orientation algorithm with bounded reachability can be based on a \clocal\ {\em coloring\/} algorithm.
Indeed, every vertex-coloring with $k$ colors induces an orientation with reachability
$O(\Delta^k)$. Towards this end, we design a \clocal\ coloring algorithm that applies techniques from \dlocal\ colorings algorithms~\cite{cole1986deterministic,goldberg1988parallel,linial,panconesi2010fast}.
 Our \clocal\  algorithm is deterministic, does not use any space between queries, performs
 $O(\Delta\cdot \log^*(n)+\Delta^2)$ probes per query, and computes a coloring
 with $O(\Delta^2\log(\Delta))$ colors. (We refer to the problem of coloring a graph by $\poly(\Delta)$ colors as \pcolor.) Our coloring algorithm yields an orientation
 whose reachability is $\Delta^{O(\Delta^2\log(\Delta))}$.
 For constant degree graphs, this implies $O(\log^*(n))$ probes to obtain an orientation with constant reachability.
 As an application of this orientation algorithm, we also  design a \clocal\ algorithm for $(\Delta+1)$-coloring.

\paragraph{Local centralized simulations of sequential algorithms}
We apply a general transformation (similarly to what was shown in~\cite{shaisoda2012}) from global
algorithms with certain properties to local algorithms.
The transformation is based on our \clocal\ orientation with bounded reachability algorithm.
As a result we get
deterministic \clocal\ algorithms for \mis\ and maximal matching (\mm), which significantly
improve over previous work~\cite{shaiics2011,shaisoda2012,shaiicalp2012}, and the first \clocal\
algorithm for coloring with $(\Delta+1)$ colors (We refer to the problem of coloring a graph by $\Delta+1$ colors as \dcolor). Compared to previous work,
for \mis\ and \mm\ the dependence on $n$ in the probe complexity is reduced from
${\rm polylog}(n)$ to $\log^*(n)$ and the space needed to store the state between queries is reduced from ${\rm
  polylog}(n)$ to zero.

\ifnum\esa=0
\paragraph{The recursive local improvement technique}
Our final result in the \clocal\ model is a $(1-\eps)$-approximation algorithm for
maximum matching (\mum). This algorithm applies an additional technique,
introduced by Nguyen and Onak~\cite{onak2008} in the design of their sublinear algorithm
for approximating the size of a maximum matching. This technique, which we refer to
as {\em Local Improvement\/} builds on
augmenting paths. Roughly speaking, it consists of recursive applications of a \clocal\ algorithm for
\mis\ to certain auxiliary graphs (defined by augmenting paths). Here we also
reduce the dependence of the probe-complexity
on $n$ from  ${\rm polylog}(n)$~\cite{shaiapprox2013} to $\log^*(n)$ and the space needed to store the state between queries
is reduced from ${\rm polylog}(n)$ to $0$.
\fi

\ifnum\esa=0
\paragraph{A deterministic \dlocal\ algorithm for approximate maximum matching}
Let $r$ denote an upper bound on the radius of the ball that contains all the probes
of our \clocal\ algorithm for approximate maximum matching ($(1-\eps)$-\mum). A
deterministic \dlocal$[r]$-algorithm for $(1-\eps)$-\mum\ follows immediately.
Indeed, $r$ is $\Delta^{O(1/\eps)}+O\left(\frac{1}{\eps^2}\right)\cdot \log^*(n)$.
In~\cite{DBLP:conf/spaa/LotkerPP08} a randomized \dlocal$[O((\log(n))/\eps^{3})]$-algorithm was presented for $(1-\eps)$-\mum. The number of rounds
in~\cite{DBLP:conf/spaa/LotkerPP08} does not depend on $\Delta$, which is, therefore,
unbounded. (See also~\cite{DBLP:journals/siamcomp/LotkerPR09,panconesi2010fast} for
constant approximation \dlocal$[\polylog (n))]$-algorithms for maximum weighted
matching.)  Hence we get an improved result when $\Delta^{O(1/\eps)} = o(\log(n))$
(and, in particular, for constant $\Delta$ and $\eps$).  Note that an
$O(1)$-approximation of a maximum matching in a $n$-node ring cannot be computed by a
\dlocal$[O(1)]$-algorithm~\cite{czygrinow2008fast,lenzen2008leveraging}.
\fi

\subsection{Comparison to Previous Work}\label{subsec:rel-work}

\paragraph{Comparison to previous (explicit) \clocal\ algorithms}
A comparison of our results with previous \clocal\ algorithms is summarized in
Table~\ref{tab:app}. The dependence on $\Delta$
 \ifnum\esa=0 and $\eps$ \fi of previous algorithms is not explicit; the dependency in Table~\ref{tab:app} is based on our understanding of these results.

\paragraph{Comparison to previous \clocal\ oracles in sublinear approximation
  algorithms}
A sublinear approximation algorithm for a certain graph parameter (e.g., the size of a minimum
vertex cover) is given probe access to the input graph and is required to output an approximation
of the graph parameter with high (constant) success probability. Many such algorithms work by
designing an {\em oracle\/} that answers queries (e.g., a query can ask: does a given vertex belong to a
fixed small vertex cover?). The sublinear approximation algorithm estimates the graph parameter
by performing (a small number of) queries to the oracle. The oracles are essentially
\clocal\ algorithms but they tend to have constant error probability, and it is not
clear how to reduce this error probability without significantly increasing their probe
complexity. Furthermore, the question of bounded space needed to store the state between queries was not an issue in the design
of these oracles, since only few queries are performed by the sublinear approximation
algorithm. Hence, they are not usually considered to be ``bona fide'' \clocal\ algorithms.
A comparison of our results and these oracles appears in Table~\ref{tab:app2}.

\paragraph{Previous \dlocal\ algorithms for \mum}
In~\cite{DBLP:conf/spaa/LotkerPP08} randomized algorithms
for finding approximately optimal matchings in both
weighted and unweighted graphs. For unweighted graphs, they
give an algorithm providing at least
$(1-\eps)$-approximation  in $O(\eps^{-3}\cdot \log(n))$
rounds for any $\eps > 0$ (w.h.p). For weighted graphs,
they give an algorithm providing at least
$(2+\eps)$-approximation in $O(\log(\eps^{-1} \cdot
\log(n)))$ rounds for any $\eps >0$ (w.h.p).
In~\cite{DBLP:journals/siamcomp/LotkerPR09} a randomized
$(4+\epsilon)$-approximation distributed algorithm for
maximum \emph{weighted} matching, whose running time is
$O(\eps^{-1}\cdot \log(\eps^{-1}) \cdot \log(n))$ for any
$\epsilon>0$ (w.h.p), where $n$ is the number of nodes in
the graph. In~\cite{panconesi2010fast} the authors solve
the weighted $b$-matching problem, which is the
generalization of the weighted matching problem where for
each vertex $v$, at most $b(v)$ edges incident to v, can be
included in the matching. For this problem they obtain a
randomized distributed algorithm with approximation
guarantee of $(6+\eps)$ , for any $\eps > 0$ that requires
$O\left( \frac{\log^3 (n)}{\eps^3}\cdot \log^2
(\Gamma)\right)$ rounds , where $\Gamma$ is the ratio
between the maximum to minimum edge weight. For weighted
matching, they devise a deterministic distributed algorithm
with the same approximation ratio that requires $O\left(
\frac{\log^4 (n)}{\eps}\cdot \log (\Gamma)\right)$.

\begin{table}[htb]
\footnotesize
\centering
\begin{tabular}{| c | c | c || c| c | c |c |}
\hline
\multicolumn{3}{|c ||}{Previous work } &  \multicolumn{2}{|c |}{Here (Deterministic)}\\
\cline{1-5}
\# Rounds & success prob. & apx. ratio & \# Rounds & apx. ratio\\
\hline\hline
$O(\eps^{-3}\cdot \log(n))$ & $1-\frac{1}{\poly(n)}$ & $1-\eps$ ~\cite{DBLP:conf/spaa/LotkerPP08}&  \multirow{4}{*}{$\Delta^{O(1/\eps)}+O\left(\frac{1}{\eps^2}\right)\cdot \log^*(n)$} & \multirow{4}{*}{$1-\eps$} \\
$O(\eps^{-1}\cdot \log\eps^{-1} \cdot \log(n))$ & $1-\frac{1}{\poly(n)}$ & $4+\eps$ ~\cite{DBLP:journals/siamcomp/LotkerPR09}& &  \\
$O( \eps^{-3} \cdot \log^3 n)$ & $1-\frac{1}{\poly(n)}$ & $6+\eps$ ~\cite{panconesi2010fast}& &  \\
$O( \eps^{-1} \cdot \log^4 n)$ & deterministic & $6+\eps$ ~\cite{panconesi2010fast}& & [Thm.~\ref{thm:distalg}] \\
\hline
\end{tabular}
\caption{\small A comparison between Previous \mum\ \dlocal\ algorithms and ours.}
   \label{tab:dist}
\end{table}

\fi 

\section{Preliminaries}
\subsection{Notations}
Let $G=(V,E)$ denote an undirected graph.  Let $n(G)$
denote the number of vertices. We denote the degree of $v$ by $\degree_G(v)$.
Let $\Delta(G)$ denote the maximum degree, i.e., $\Delta(G)
\eqdf \max_{v\in V}\{\degree_G(v)\}$. The length of a path
equals the number of edges along the path. We denote the
length of a path $p$ by $|p|$. For $u,v \in V$ let
$\dist_G(u,v)$ denote the length of the shortest path
between $u$ and $v$ in the graph $G$. The ball of radius
$r$ centered at $v$ in the graph $G$ is defined by
\[
    B^G_r(v) \triangleq \{u \in V \mid \dist_G(v,u) \leq r\}\:.
\]
If the graph $G$ is clear from the context, we may drop it
from the notation, e.g., we simply write $n,m,\degree(v)$, or
$\Delta$.

For $k \in \NN^+$ and $n >0$, let $\log^{(k)} (n)$ denote
the $k$th iterated logarithm of $n$. Note that $\log ^{(0)}
(n) \triangleq n$ and if $\log^{(i)} (n)=0$, we define
$\log^{(j)} (n) = 0$, for every $j>i$.
For $n \geq 1$, define $\log^{*} (n)\triangleq \min \{i:
\log^{(i)}(n) \leq 1\}$.

A subset $I\subseteq V$ is an \emph{independent set} if no two vertices in $I$ are an
edge in $E$. An independent set $I$ is \emph{maximal} if $I\cup\{v\}$ is not an
independent set for every $v\in V\setminus I$. We use \mis\ as an abbreviation of a
maximal independent set.

A subset $M\subseteq E$ is a matching if no two edges in $M$ share an endpoint.  Let
$M^*$ denote a maximum cardinality matching of $G$.  We say that a matching $M$ is a
$(1-\eps)$-approximate maximum matching if
\[
    |M| \geq (1-\eps)\cdot|M^*|\:.
\]
Let $w(e)$ denote the weight of an edge $e\in E$. The weight of a subset $F\subseteq
E$ is $\sum_{e\in F} w(e)$ and is denoted by $w(F)$.  Let $M_w^*$ denote a maximum
weight matching of $G$. A matching $M$ is a $(1-\eps)$-approximate maximum weight matching
if $w(M) \geq (1-\eps)\cdot w(M_w^*)$. We abbreviate the terms maximum cardinality
matching and maximum weight matching by \mcm\ and \mwm, respectively.
\subsection{The \dlocall\ Model}\label{sec:dlocal-def}
The model of local distributed computation is a classical
model (e.g.,~\cite{linial,pelegbook,dsurvey}). A
distributed computation takes place in an undirected
labeled graph $G=(V,E)$.  In a labeled graph vertices have
unique identifies (IDs). The neighbors of each vertex $v$
are numbered from $1$ to $\deg(v)$ in an arbitrary but
fixed manner. Ports are used to point to the neighbors of
$v$; the $i$th port points to the $i$th neighbor.  Each
vertex in the labeled graph models a processor, and
communication is possible only between neighboring
processors.  All processors execute the same algorithm.
Initially, every $v \in V$ is input a local input.  The
computation is done in $r \in \NN$ synchronous rounds as
follows.  In every round: (1)~every processor receives a
message from each neighbor, (2)~every processor performs a
computation based on its local input and the messages
received from its neighbors, (3)~every processor sends a
message to each neighbor.  We assume that a message sent in
the end of round $i$ is received in the beginning of round
$i+1$.  After the $r$th round, every processor computes a
local output.

The following assumptions are made in the \dlocal\ model:
\begin{inparaenum}[(1)]
\item The local input to each vertex $v$ includes the ID of $v$, the degree of the
  vertex $v$, the maximum degree $\Delta$, the number of vertices $n$, and the ports
  of $v$ to its neighbors.
\item The IDs are distinct and bounded by a polynomial in $n$.
\item The length of the messages sent in each round is not bounded.
\end{inparaenum}

We say that a distributed algorithm is a \emph{$\dlocal[r]$-algorithm} if the number
of communication rounds is $r$.  Strictly speaking, a distributed algorithm is
considered {\em local\/} if $r$ is bounded by a constant.  We say that a
$\dlocal[r]$-algorithm is \emph{almost local} if $r=O(\log^*(n))$.  When it is obvious
from the context we refer to an almost \dlocal\ algorithm simply by a \dlocal\
algorithm.


\subsection{The \clocall\ Model}\label{sec:cmodel}
In this section we present the model of centralized local computations that was defined
in~\cite{shaiics2011}.  The presentation focuses on problems over labeled graphs
(i.e., maximal independent set and maximum matching).

\paragraph{Probes}
In the \clocal\ model, access to the labeled graph is
limited to probes.  A \emph{probe} is a pair $(v,i)$ that
asks ``who is the $i$th neighbor of $v$?''.  The answer to
a probe $(v,i)$ is as follows.  (1)~If $\degree_G(v)<i$,
then the answer is ``null''. (2)~If $\degree_G(v)\geq i$,
then the answer is the (ID of) vertex $u$ that is pointed
to by the $i$th port of $v$. For simplicity, we assume that
the answer also contains the port number $j$ such that $v$
is the $j$th neighbor of $u$. (This assumption reduces the
number of probes by at most a factor of $\Delta$.)

\paragraph{Online Property of \clocal-algorithms}
The set of solutions for problem $\Pi$ over a labeled graph $G$ is denoted by
$\sol(G,\Pi)$.  A deterministic \clocal-algorithm \alg\ for problem $\Pi$ over
labeled graphs is defined as follows.  The input for the algorithm consists of three
parts: (1)~access to a labeled graph $G$ via probes, (2)~the number of vertices $n$
and the maximum degree $\Delta$ of the graph $G$, and (3)~a sequence
$\{q_i\}_{i=1}^N$ of queries.  Each query $q_i$ is a request for an evaluation of
$f(q_i)$ where $f\in \sol(G,\Pi)$. Let $y_i$ denote the output of \alg\ to query
$q_i$. We view algorithm \alg\ as an online algorithm because it must output $y_i$ without any
knowledge of subsequent queries.

A \clocal-algorithm $\alg$ for $\Pi$ must satisfy the following condition, called
\emph{consistency},
\begin{equation}\label{eq:const-def}
    \exists f \in  \sol(G,\Pi) \mbox{ s.t. }~\forall N \in \NN ~~\forall \{q_i\}_{i=1}^N ~~\forall i ~:~y_i = f(q_i)\:.
\end{equation}

\paragraph{Resources and Performance Measures}
The main performance measure is the \emph{maximum number of
probes} that the \clocal-algorithm performs per query.
We consider an additional measure called \emph{probe
radius}.  The probe radius of a \clocal-algorithm $C$ is
$r$ if, for every query $q$, all the probes that algorithm
$C$ performs in $G$ are contained in the ball of radius $r$
centered at $q$.

\paragraph{Stateless Algorithms}
A \emph{state} of a \clocal-algorithm is the maximum number of
bits stored between consecutive queries.
A \clocal-algorithm is \emph{stateless} if the algorithm
does not store any information between queries. In
particular, a stateless algorithm does not store previous
queries, answers to previous probes, or answers given to
previous queries.\footnote{We remark that
in~\cite{shaiics2011} no distinction was made between the
  space needed to answer a query and the space needed to store the state between
  queries. Our approach is different and follows the \dlocal\ model in which one does
  not count the space and running time of the vertices during the execution of the
  distributed algorithm. Hence, we ignore the space and running time of the
  \clocal-algorithm during the processing of a query.}
  In this paper all our \clocal-algorithms are stateless.

\paragraph{Example} Consider the problem of computing a maximal independent set.
The \clocal-algorithm is input a sequence of queries, each
of which is a vertex. The algorithm outputs whether $q_i$
is in $I$, for some maximal independent set $I\subseteq V$.
Consistency in this example means that the algorithm has to
satisfy this specification even though it does not probe
all of $G$, and obviously does not store the maximal
independent set $I$.  Moreover, a stateless algorithm does
not even remember the answers it gave to previous queries.
Note that if a vertex is queried twice, then the algorithm
must return the same answer.  Similarly, if two queries are
neighbors, then the algorithm may not answer that both are
in the independent set.  If all vertices are queried, then
the answers constitute the maximal independent set $I$.

\subsection{Simulation of \clocall\ by
\dlocall}\label{sec:simu}

Based on an observation made in~\cite{parnasron} in a slightly different setting,
\clocal-algorithms can simulate \dlocal-algorithms.
In this section we consider simulations in the
converse direction.

The following definition considers \clocal-algorithms whose queries are vertices of a
graph. The definition can be easily extended to edge queries.
\begin{definition}
  A \dlocal-algorithm $D$ simulates a \clocal-algorithm $C$ if, for every vertex $v$,
  the local output of $D$ in vertex $v$ equals the answer that algorithm $C$ computes
  for the query $v$.
\end{definition}

The following proposition states that \clocal-algorithm can be simulated by a
$\dlocal[r]$ algorithm provided that the probe radius is $r$. Message lengths grow at
a rate of $O(\Delta^{r+1} \cdot \log n)$ as information (e.g., IDs and existence of
edges) is accumulated.
\begin{proposition}\label{prop:simul}
  Every stateless deterministic \clocal-algorithm $C$ whose probe radius is at most $r$ can be
  simulated by a deterministic $\dlocal[r]$-algorithm $D$.
\end{proposition}
\begin{proof}
  The distributed algorithm $D$ collects, for every $v$, all the information in the
  ball of radius $r$ centered at $v$. (This information includes the IDs of the
  vertices in the ball and the edges between them.)

  After this information is collected, the vertex $v$ locally runs the
  \clocal-algorithm $C$ with the query $v$. Because algorithm $C$ is stateless, the
  vertex has all the information required to answer every probe of $C$.
\end{proof}

Proposition~\ref{prop:simul} suggests a design methodology for distributed
algorithms. For example, suppose that we wish to design a distributed algorithm for
maximum matching. We begin by designing a \clocal-algorithm $C$ for computing a
maximum matching. Let $r$ denote the probe radius of the \clocal-algorithm $C$. The
proposition tells us that we can compute the same matching (that is computed by $C$)
by a distributed $r$-round algorithm.

\section{Acyclic Orientation with Bounded Radius (\oradd)} \label{sec:obr}
In this section we define the problem of \emph{Acyclic Orientation with Bounded Radius} (\orad).
We then design a \clocal\ algorithm for \orad\ based on
vertex coloring.

\paragraph{Definitions}
An \emph{orientation} of an undirected graph $G=(V,E)$ is a
directed graph $H=(V,A)$, where $\{u,v\}\in E$ if and only
if $(u,v)\in A$ or $(v,u)\in A$ but not both. An
orientation $H$ is \emph{acyclic} if there are no directed
closed paths in $H$. The \emph{radius} of an acyclic
orientation $H$ is the length of the longest directed path
in $H$. We denote the radius of an orientation by
$\rad(H)$.  In the problem of acyclic orientation with
bounded radius (\orad), the input is an undirected graph.
The output is an orientation $H$ of $G$ that is acyclic.
The goal is to compute an acyclic orientation $H$ of $G$
that minimizes $\rad(H)$.

As in~\cite{EMR14}, an acyclic orientation is induced by a vertex coloring. Previous
works obtain an acyclic orientation by random vertex
ranking~\cite{onak2008,yoshida2009improved,shaisoda2012,shaiicalp2012,shaiapprox2013}.

\begin{proposition}[Orientation via coloring]\label{prop:color2ori}
  Every coloring by $c$ colors induces an acyclic  orientation $H$ with
  \[
    \rad(H) \leq c-1.
  \]
\end{proposition}
\begin{proof}
  Direct each edge from a high color to a low color.  By monotonicity the orientation
  is acyclic.  Every directed path has at most $c$ vertices, and hence the
  radius is bounded as required.
\end{proof}

Many distributed coloring algorithms find a vertex coloring
in $O(\log^*(n) + \poly(\Delta))$ rounds (giving us the
same upper bound on the probe-radius of the corresponding
\clocal-algorithm) and use $\poly(\Delta)$ colors~(see, for
example,
\cite{barenboim2009distributed,linial,cole1986deterministic,panconesi2001some,kuhn2009weak}).
For concreteness, in this paper,  we employ a \clocal\
simulation of a distributed vertex coloring algorithm with
$O(\log^*(n) + \poly(\Delta))$ rounds that uses
$\poly(\Delta)$ colors.

We remark that a
\clocal-algorithm
for \orad\ simply
computes, for every vertex $v$ and every port $i$, whether
the edge incident to $v$ at port $i$ is an incoming edge or
an outgoing edge in the orientation.

\section{A \clocall-Algorithm for \miss}\label{sec:lin}
In this section we briefly describe a \clocal-algorithm for
the \mis\ problem.  The algorithm is a special case of a
more general technique of designing \clocal-algorithms from
``greedy'' (global) sequential
algorithms~(see~\cite{EMR14,shaiicalp2012}).

Suppose we wish to compute a maximal independent set $\mis$ of a graph $G=(V,E)$.
The greedy algorithm proceeds by scanning the vertices in some ordering $\sigma$. A
vertex $v$ is added to the \mis\ if none of its neighbors that appear before $v$ in
$\sigma$ have been added to the \mis.  Let $\mis_\sigma$ denote the \mis\ that is
computed by the greedy algorithm if the vertices are scanned by the ordering
$\sigma$.

Every acyclic orientation $H=(V,A)$ of $G$ induces a partial order $P_H$ simply by
considering the transitive closure of $H$.  The key observation is that
$\mis_{\sigma}= \mis_{\tau}$ for every two linear orderings $\sigma$ and $\tau$ that
are linear extensions of $P_H$. Let $\mis_{P_H}$ denote the \mis\ that corresponds to
$\mis_{\sigma}$ for linear extensions $\sigma$ of $P_H$.

 A \clocal-algorithm can compute
$\mis_{P_H}$ as follows. Given $v$, the algorithm performs
a directed DFS from $v$ according to the directed edges
$A$. When the DFS backtracks from a node $u$, it adds $u$
to $\mis_{P_H}$ if none of the descendants of $u$ in
the DFS tree are in $\mis_{P_H}$.

\medskip\noindent
We conclude with the following lemma that summarizes the above description.
\begin{lemma}\label{lemma:simul}
  Let $\ao$ denote a stateless \clocal-algorithm that computes an acyclic orientation
  $H=(V,A)$ of a graph $G=(V,E)$.  Let $r$ denote the probe radius of $\ao$.
  Then, there exists a stateless \clocal-algorithm for \mis\ whose probe radius is at most
  $r+\rad(H)$.
\end{lemma}
Let \lmis\ denote the \clocal\ \mis-algorithm in Lemma~\ref{lemma:simul}. Let
$\lmis(G,v)$ denote the Boolean predicate that indicates if $v$ is in the \mis\ of
$G$ computed by Algorithm \lmis.
\section{A  \clocall\ Approximate \mcmm\ Algorithm}\label{sec:clocal mcm}
In this section we present a stateless deterministic
\clocal\ algorithm that computes a $(1-\eps)$-approximation
of a maximum cardinality matching. The algorithm is based on a
\clocal-algorithm for maximal independent set (see
Lemma~\ref{lemma:simul}) and on the local improvement
technique of Nguyen and Onak~\cite{onak2008}.

\paragraph{Terminology and Notation}
Let $M$ be a matching in $G=(V,E)$.
A vertex $v \in V$ is \emph{$M$-free} if $v$ is not an endpoint of an edge in $M$.
A simple path is \emph{$M$-alternating} if it consists of edges drawn alternately from $M$ and from $E \setminus M$.
A path is \emph{$M$-augmenting} if it is $M$-alternating and if both of the path's endpoints are $M$-free vertices.
Note that the length of an augmenting path must be odd.
The set of edges in a path $p$ is denoted by $E(p)$, and the set of edges in a
collection $P$ of paths is denoted by $E(P)$.
Let $A\oplus B$ denote the symmetric difference of the sets $A$ and $B$.

\paragraph{Description of The Global Algorithm}
Similarly to~\cite{DBLP:conf/spaa/LotkerPP08,onak2008, shaiapprox2013} our local
algorithm simulates the global algorithm listed as Algorithm~\ref{alg:global mcm}. This
global algorithm builds on the next two lemmas of Hopcroft and
Karp~\cite{hopcroft1973n}, and Nguyen and Onak~\cite{onak2008}, respectively.

\begin{lemma}[\cite{hopcroft1973n}]
  Let $M$ be a matching in a graph $G$. Let $k$ denote the length of a shortest
  $M$-augmenting path. Let $P^*$ be a maximal set of vertex disjoint $M$-augmenting
  paths of length $k$.  Then, $(M\oplus E(P^*))$ is a matching and the length of every $(M\oplus E(P^*))$-augmenting path
  is at least $k+2$.
  \label{lemma:hopcroft-karp}
\end{lemma}

\begin{lemma}[{\cite[Lemma~6]{onak2008}}]\label{lemma:apxmatch}
  Let $M^*$ be a maximum matching and $M$ be a matching in a graph $G$. Let $2k+1$
  denote the length of a shortest $M$-augmenting path.  Then
  \[
    |M| \geq \frac{k}{k+1}\cdot |M^*|\:.
  \]
\end{lemma}

\begin{algorithm}[H]\scriptsize
\caption{$\text{Global-APX-MCM}(G,\eps)$.}
\label{alg:global mcm}
\begin{algorithmic}[1]
\Require A graph $G=(V,E)$ and $0<\eps <1$.
\Ensure  A $(1-\eps)$-approximate matching
\State $M_0\gets \emptyset$.
\State $k\gets \lceil \frac{1}{\eps} \rceil$.
\For {$i=0$ to $k$}
  \State $P_{i+1} \gets \{ p \mid \text{$p$ is an
      $M_{i}$-augmenting path}, |p|=2i+1\}$.
  \State $P^*_{i+1}\subseteq P_{i+1}$ is a maximal vertex
      disjoint subset of paths.
  \State $M_{i+1}\triangleq M_{i} \oplus E(P^*_{i+1})$.
  \EndFor
\State \textbf{Return} $M_{k+1}$.
\end{algorithmic}
\end{algorithm}
%
%
\begin{algorithm}[H]\scriptsize
\caption{$\text{Global-APX-MCM'}(G,\eps)$.}
\label{alg:global mcm'}
\begin{algorithmic}[1]
\Require A graph $G=(V,E)$ and $0<\eps <1$.
\Ensure  A $(1-\eps)$-approximate matching
\State $M_0\gets \emptyset$.
\State $k\gets \lceil \frac{1}{\eps} \rceil$.
\For {$i=0$ to $k$}
  \State Construct the intersection graph $H_i$ over $P_i$.
  \State $P^*_{i+1} \gets \mis (H_i)$.
  \State $M_{i+1}\triangleq M_{i} \oplus E(P^*_{i+1})$.
  \EndFor
\State \textbf{Return} $M_{k+1}$.
\end{algorithmic}
\end{algorithm}

Algorithm~\ref{alg:global mcm} is given as input a graph $G$ and an approximation parameter
$\eps\in (0,1)$. The algorithm works in $k$
iterations, where $k=\lceil \frac{1}{\eps}\rceil$.  Initially, $M_0=\emptyset$. The
invariant of the algorithm is that $M_i$ is a matching, every augmenting path of
which has length at least $2i+1$.  Given $M_{i}$, a new matching $M_{i+1}$ is
computed as follows.  Let $P_{i+1}$ denote the set of shortest $M_{i}$-augmenting
paths. Let $P^*_{i+1} \subseteq P_{i+1}$ denote a maximal subset of vertex disjoint
paths.  Define $M_{i+1}\triangleq M_{i} \oplus E(P^*_{i+1})$.  By
Lemmas~\ref{lemma:hopcroft-karp} and~\ref{lemma:apxmatch}, we obtain the following
result.

\begin{theorem}
  The matching $M_{k+1}$ computed by Algorithm~\ref{alg:global mcm} is a
  $(1-\eps)$-approximation of a maximum matching.
\end{theorem}

\paragraph{The intersection graph}
Define the intersection graph $H_i=(P_i,C_i)$ as follows.  The set of nodes $P_i$ is
the set of $M_{i-1}$-augmenting paths of length $2i-1$.  We connect two paths $p$ and
$q$ in $P_i$ by an edge $(p,q) \in C_i$ if $p$ and $q$ intersect (i.e., share a
vertex in $V$).  Note that $H_1$ is the line graph of $G$ and that $M_1$ is simply a
maximal matching in $G$.  Observe that $P^*_i$ as defined above is a maximal
independent set in $H_i$.  Thus, iteration $i$ of the global algorithm can be
conceptualized by the following steps (see Algorithm~\ref{alg:global mcm'}): construct the intersection graph $H_i$,
compute a maximal independent set $P^*_i$ in $H_i$, and augment the matching by
$M_{i}\triangleq M_{i-1} \oplus(E(P^*_i))$.

%

\paragraph{Implementation by a stateless deterministic \clocal\ Algorithm}
The recursive local improvement technique in~\cite[Section
3.3]{onak2008} simulates the global algorithm. It is based
on a recursive oracle $\oracle_i$.  The input to oracle
$\oracle_i$ is an edge $e \in E$, and the output is a bit
that indicates whether $e \in M_i$. Oracle $\oracle_i$
proceeds by computing two bits $\tau$ and $\rho$ (see
Algorithm~\ref{alg:Oi}).  The bit $\tau$ indicates whether
$e\in M_{i-1}$, and is computed by invoking oracle
$\oracle_{i-1}$.  The bit $\rho$ indicates whether $e\in
E(P^*_i)$ (where $P^*_i$ is an \mis\ in $H_{i-1}$).  Oracle
$\oracle_i$ returns $\tau \oplus \rho$ because
$M_i=M_{i-1}\oplus E(P^*_i)$.

We determine whether $e\in E(P^*_i)$ by running the \clocal-algorithm $\proca_i$ over
$H_i$ (see Algorithm~\ref{alg:Ai}).  Note that $\proca_1$ simply computes a maximal
matching (i.e., a maximal independent set of the line graph $H_1$ of $G$).  The main
difficulty we need to address is how to simulate the construction of $H_i$ and probes
to vertices in $H_i$. We answer the question whether $e \in E(P_i^*)$ by executing
the following steps: (1)~Listing: construct the set $P_i(e)\triangleq \{p\in P_i \mid
e\in E(p)\}$. Note that $e\in E(P^*_i)$ if and only if $P_i(e)\cap P^*_i\neq
\emptyset$. (2)~\mis-step: for each $p\in P_i(e)$, input the query $p$ to an
\mis-algorithm for $H_i$ to test whether $p\in P^*_i$. If an affirmative answer is
given to one of these queries, then we conclude that $e\in E(P_i^*)$. We now
elaborate on how the listing step and the \mis-step are carried out by a
\clocal-algorithm.

The listing of all the paths in $P_i(e)$ uses two preprocessing steps: (1)~Find the
balls of radius $2i-1$ in $G$ centered at the endpoints of $e$. (2)~Check if $e'\in
M_{i-1}$ for each edge $e'$ incident to vertices in the balls. We can then
exhaustively check for each path $p$ of length $2i-1$ that contains $e$ whether $p$
is an $M_{i-1}$-augmenting path.

The \mis-step answers a query $p\in P^*_i$ by simulating
the \mis\ \clocal-algorithm over $H_i$. The \mis-algorithm
probes $H_i$. A probe to $H_i$ consists of an
$M_{i-1}$-augmenting path $q$ and a port number. We suggest
to implement this probe by probing all the neighbors of $q$
in $H_i$ (so the port number does not influence the first
part of implementing a probe). See Algorithm~\ref{alg:prob
Hi}. As in the listing step, a probe $q$ in $H_i$ can be
obtained by (1)~finding the balls in $G$ of radius $2i-1$
centered at endpoints of edges in $E(q)$, and (2)~finding
out which edges within these balls are in $M_{i-1}$. The
first two steps enable us to list all of the neighbors of
$q$ in $H_i$ (i.e., the $M_{i-1}$-augmenting paths that
intersect $q$). These neighbors are ordered (e.g., by
lexicographic order of the node IDs along the path). If the
probe asks for the neighbor of $q$ in port $i$, then the
implementation of the probe returns the $i$th neighbor of
$q$ in the ordering.

By combining the recursive local improvement technique with
our deterministic stateless \clocal\ \mis-algorithm, we
obtain a deterministic stateless \clocal-algorithm that
computes a $(1-\eps)$-approximation for maximum matching.
The algorithm is invoked by calling the oracle
$\oracle_{k+1}$. The next lemma can be proved by induction.
\begin{lemma}
  The oracle $\oracle_{i}(e)$ computes whether $e\in M_{i}$.
\end{lemma}
%
\begin{algorithm}[H]\scriptsize
\caption{$\oracle_{i}(e)$ - a recursive oracle for
membership in the approximate matching.} \label{alg:Oi}
\begin{algorithmic}[1]
\Require A query $e \in E$.
\Ensure Is $e$ an edge in the matching $M_i$?
  \State If $i=0$ then return false.
  \State $\tau\gets \oracle_{i-1}(e)$.
  \State $\rho\gets \proca_{i}(e)$.
  \State \textbf{Return} $\tau \oplus \rho$.
\end{algorithmic}
\end{algorithm}
%
\begin{algorithm}[H]\scriptsize
\caption{$\proca_i(e=(u,v))$ - a procedure for checking
membership of an edge $e$ in one
  of the paths in $P^*_i$.} \label{alg:Ai}
\begin{algorithmic}[1]
\Require An edge $e \in E$.
\Ensure Does $e$ belong to a path $p\in P^*_i$?
\State \textbf{Listing step:} \Comment Compute all shortest $M_{i-1}$-augmenting
paths that contain $e$.
  \State \quad $B_u \gets BFS_G(u)$ with depth $2i-1$.
  \State \quad $B_v \gets BFS_G(v)$ with depth $2i-1$.
  \State \quad For every edge $e'$ in the subgraph of $G$ induced by $B_u\cup B_v$: $\chi_{e'}\gets \oracle_{i-1}(e')$.
  \State \quad $P_i(e)\gets$ all $M_{i-1}$-augmenting paths of length $2i-1$ that contain $e$.
\State \textbf{\mis-step:} \Comment Check if one of the augmenting paths is in $P^*_i$.
  \State \quad For every $p\in P_i(e)$: If $\lmis(H_i,p)$ \textbf{Return} true.
  \State \quad \textbf{Return} false.
\end{algorithmic}
\end{algorithm}
%
\begin{algorithm}[H]\scriptsize
  \caption{$\textit{probe}(i,p)$ - simulation of a probe to the intersection graph
    $H_i$ via probes to $G$.} \label{alg:prob Hi}
\begin{algorithmic}[1]
\Require A path $p \in P_i$ and the ability to probe $G$.
\Ensure The set of $M_{i-1}$-augmenting paths of length $2i-1$ that intersect $p$.
\State For every $v\in p$ do
  \State \quad $B_v \gets BFS_G(v)$ with depth $2i-1$.
  \State \quad For every edge $e'\in B_v$: $\chi_e\gets \oracle_{i-1}(e)$. \Comment determine whether the path is alternating and whether the endpoints are $M_{i-1}$-free.
  \State \quad $P_i(v)\gets$ all $M_{i-1}$-augmenting paths of length $2i-1$ that contain $v$.
\item \textbf{Return} $\bigcup_{v\in p} P_i(v)$. 
\end{algorithmic}
\end{algorithm}

\section{A  \dlocall\ Approximate \mcmm\ Algorithm}\label{sec:dlocal mcm}
In this section, we present a \dlocal-algorithm that
computes a $(1-\eps)$-approximate maximum cardinality matching.  The
algorithm is based on collecting information from balls and
then simulating the \clocal\ algorithm presented in
Section~\ref{sec:clocal mcm}.
\begin{theorem}\label{thm:distalg}
  There is a deterministic $\dlocal[\Delta^{O(1/\eps)} +
  O\left(\frac{1}{\eps^2}\right) \cdot\log^*(n)]$-algorithm for computing a
  $(1-\eps)$-approximate \mum.
\end{theorem}
\begin{proof}
  The proof of the theorem is based on the simulation of a \clocal-algorithm by a
  \dlocal-algorithm, as summarized in Section~\ref{sec:simu}.  In
  Lemma~\ref{lemma:radius} we prove that the probes are restricted to a ball of radius
  $\Delta^{O(1/\eps)} + O\left(\frac{1}{\eps^2}\right) \cdot\log^*(n)$, and the
  theorem follows.
\end{proof}
\begin{lemma}\label{lemma:radius}
The probe radius of the \clocal-algorithm $\oracle_{1+\lceil 1/\eps \rceil}$ is
  \[
    r = \Delta^{O(1/\eps)}+O\left(\frac{1}{\eps^2}\right)\cdot \log^*(n)\:.
  \]
\end{lemma}
\begin{proof}
  Consider a graph $G'$ and a \clocal-algorithm $A$ that probes $G'$.  Let
  $r_{G'}(A)$ denote the probe radius of algorithm $A$ with respect to the graph
  $G'$.

  The description of the oracle $\oracle_i$ implies that the probe radius $r_G(\oracle_i)$ satisfies
  the following recurrence:
  \begin{align*}
    r_G(\oracle_i)&=
    \begin{cases}
      0 & \text{if $i=0$},\\
      r_G(\proca_1) & \text{if $i=1$},\\
      \max\{r_G(\oracle_{i-1}), r_G(\proca_i)\}& \text{if $i\geq 2$.}
\end{cases}
\end{align*}

The description of the procedure $\proca_i$ implies that the probe radius $r_G(\proca_i)$ satisfies
the following recurrence:
  \begin{align*}
    r_{G}(\proca_i)&\leq
    \max\{2i+r_G(\oracle_{i-1}), 2i-1+r_G(\lmis(H_i))\}
\end{align*}

We bound the probe radius $r_G(\lmis(H_i))$ by composing the
radius $r_{H_i}(\lmis(H_i))$ with the increase in radius
incurred by the simulation of probes to $H_i$ by probes to
$G$.  Recall that the $\lmis$-algorithm is based on a
deterministic coloring algorithm $C$.  We denote the number
of colors used by $C$ to color a graph $G'$ by $|C(G')|$.

The \mis-algorithm orients the edges by coloring the vertices. The
radius of the orientation is at most the number of colors. It follows that
\begin{align*}
  r_{H_i}(\lmis(H_i)) &\leq r_{H_i}(C(H_i)) + |C(H_i)|.
\end{align*}
The simulation of probes to $H_i$ requires an increase in
the radius by a factor of $2i-1$ in addition to the radius
of the probes. Hence,
\begin{align*}
  r_{G}(\lmis(H_i)) &\leq (2i-1) \cdot   r_{H_i}(\lmis(H_i)) + r_G(\textit{probe}(i,p)).
\end{align*}

Many distributed coloring algorithms find a vertex coloring in $O(\log^*(n) +
\poly(\Delta))$ rounds (giving us the same upper bound on the probe-radius of the
corresponding \clocal-algorithm) and use $\poly(\Delta)$ colors~(see,
for example,
\cite{barenboim2009distributed,linial,cole1986deterministic,panconesi2001some,kuhn2009weak}).
Plugging these parameters in the recurrences yields
\begin{align*}
  r_G(\oracle_i) &\leq 2i +  r_G(\lmis(H_i))\\
&\leq 2i\cdot (1+  r_{H_i}(\lmis(H_i))) + r_G(\textit{probe}(i,p))\\
&\leq r_G(\oracle_{i-1}) + O\Big(i\cdot \log^* (n(H_i)) + \poly(\Delta(H_i))\Big),
\end{align*}
Since $\Delta(H_i) \leq (2i)^2 \Delta^{2i-1}$ and $n(H_i) \leq n^{2i}$ we get that
  \begin{align*}
    r_G(\oracle_k) & \leq \sum_{i=1}^{k} O\left(i\cdot \log^*(n) + \poly((2i)^2\cdot \Delta^{2i})\right) \\
    & = O(k^2\cdot \log^*(n)) + \Delta^{O(k)}.
  \end{align*}
  Let $k=1+\lceil \frac {1}{\eps}\rceil$, and the lemma follows.
\end{proof}

\section{A  \clocall\ Approximate \mwmm\ Algorithm}\label{sec:clocal mwm}
In this section we present a deterministic stateless \clocal-algorithm that computes
a $(1-\eps)$-approximation of a maximum weighted matching.  The algorithm is based on
the sublinear approximation algorithm for weighted
matching~\cite{onakthesis,onak2008}; we replace the randomized \mis-algorithm by our
deterministic \mis-algorithm.  The pseudo-code is listed in the
Appendix.

\paragraph{Terminology and Notation}
In addition to the terminology and notation used in the unweighted case, we define
the following terms. For a matching $M$ and an alternating path $p$, the \emph{gain}
of $p$ is defined by
\[
    g_M(p) \eqdf w(p\setminus M) - w(p \cap M)\:.
\]
We say that an $M$-alternating path $p$ is \emph{$M$-augmenting} if: (1)~$p$ is a
simple path or a simple cycle, (2)~$M \oplus E(p)$ is a matching, and (3)~$g_M(p) >
0$. We say that a path $p$ is \emph{$(M,[1,k])$-augmenting} if $p$ is $M$-augmenting
and $|E(p) \setminus M| \leq k$. Note that an $(M,[1,k])$-augmenting path may contain
at most $2k+1$ edges.

\paragraph{Preprocessing and Discretization of Weights}
We assume that the edge weights are positive as nonpositive weight edges do not
contribute to the weight of the matching. We also assume that the maximum edge weight
is known to all the vertices. By normalizing the weights, we obtain that the edge
weights are in the interval $(0,1]$. Note, that at least one edge has weight $1$.  As
we are interested in a $(1-\eps)$-approximation, we preprocess the edge weights by
ignoring lightweight edges and rounding down weights as follows: (1)~An edge $e$ is
\emph{lightweight} if $w(e)< \eps/n$. The contribution of the lightweight edges to
any matching is at most $\eps/2$. It follows that ignoring lightweight edges decreases the
approximation ratio by at most a factor of $(1-\eps/2)$. (2)~We round down the edge
weights to the nearest integer power of $(1-\eps/2)$. Let $w(e)$ denote the original
edge weights and let $w'(e)$ denote the rounded down weights. Therefore,
$w(e)\cdot (1-\eps/2) < w'(e)\leq w(e)$. It follows that, for every matching $M$, we
have $w'(M)\geq (1-\eps/2)\cdot w(M)$. The combined effect of ignoring lightweight
edges and discretization of edge weights decreases the approximation factor by at
most a factor of $(1-\eps)$. We therefore assume, without loss of generality, that
the edge weights $w(e)$ are integer powers of $(1-\eps/2)$ in the interval $[\eps/n,1]$.
Let
$$\wmin(\eps)\triangleq\max\{\eps/n,\min_e w(e)\}.$$
In particular, if $\wmin(\eps)<1$ (i.e., the weighted case), then there are at most
\[
W\triangleq
\Theta\left(\frac{1}{\eps}\cdot \log \left(\frac{1}{\wmin(\eps )}\right)\right)
\]
distinct weights.
This implies that the set of all possible gains achievable by $(M,[1,k])$-augmenting paths
is bounded by $T_k\triangleq\Theta(W^{2k+1})$. We denote the set of $T_k$ possible
gains by $\{g_1,\ldots, g_{T_k}\}$, where $g_i> g_{i+1}$.

\paragraph{Description of The Global Algorithm}
The starting point is the global algorithm of Pettie and
Sanders~\cite{DBLP:journals/ipl/PettieS04} for
approximating an \mwm. Onak~\cite{onakthesis} suggested an
implementation of this algorithm that is amenable to
localization (See Algorithm~\ref{alg:global mwm}
and~\ref{alg:global mwm'} in the Appendix).  The main
difference between the algorithms for the weighted case and
the unweighted case is that the maximum length of the
augmenting paths (and cycles) does not grow; instead,
during every step, $(M,[1,k])$-augmenting paths are used.
In~\cite{DBLP:journals/ipl/PettieS04}, a maximal  set of
augmenting paths is computed by greedily adding augmenting
paths in decreasing gain order. Discretization of edge
weights enables one to simulate this greedy procedure by
listing the augmenting paths in decreasing gain
order~\cite{onakthesis}.

\paragraph{Algorithm Notation} The global algorithm uses the following notation.  The
algorithm computes a sequence of matchings $M_{i,j}$ that are doubly indexed (where
$i\in[1,L]$ and $j\in [1,T_k]$).  We denote the initial empty matching by $M_{1,0}$.
These matchings are ordered in the lexicographic ordering of their indexes, and
$M_{\prev(i,j)}$ denotes the predecessor of $M_{i,j}$.  Let $P_{i,j}$ denote the set
of $(M_{\prev(i,j)},[1,k])$-augmenting paths whose gain is $g_j$. Let $H_{i,j}$ denote
the intersection graph over $P_{i,j}$ with edges between paths whenever the paths
share a vertex. Let $G_k$ denote the intersection graph over all paths of length at
most $2k+1$ in $G$. Note that each $H_{i,j}$ is the subgraph of $G_k$ induced by
$P_{i,j}$.  Hence, a vertex coloring of $G_k$ is also a vertex coloring of
$H_{i,j}$. Let $P^*_{i,j}$ denote a maximal independent set in $H_{i,j}$.

\paragraph{Implementation by a Stateless Deterministic \clocal\ Algorithm}
The \clocal\ implementation of the global algorithm is an
adaptation of the \clocal-algorithm from
Section~\ref{sec:clocal mcm}.  The oracle $\oracle_{i,j}$
is doubly indexed and so is the procedure $\proca_{i,j}$.
The algorithm is invoked by calling the oracle
$\oracle_{L,T_k}$. The next lemma can be proved by
induction.
\begin{lemma}
  The oracle $\oracle_{i,j}(e)$ computes whether $e\in M_{i,j}$.
\end{lemma}

\medskip
\noindent
The proof of the following theorem is based on the proof of Theorem
2.4.7 in~\cite{onakthesis}.
\begin{theorem}
  Algorithm~\ref{alg:global mwm'} computes a $(1-\eps)$-approximate maximum weighted matching.
\end{theorem}

\section{A  \dlocall\ Approximate \mwmm\ Algorithm}
In this section, we present a \dlocal-algorithm that
computes a $(1-\eps)$-approximate weighted matching. The
algorithm is based on the same design methodology as in
Section~\ref{sec:dlocal mcm}. Namely, we bound the probe
radius of the \clocal-algorithm for \mwm\ (see
Lemma~\ref{lemma:radius weight}) and apply the simulation
technique (see Section~\ref{sec:simu}).

\begin{theorem}\label{thm:distalg mwm}
  There is a deterministic
$\dlocal[r ]$
-algorithm for computing a
  $(1-\eps)$-approximate \mwm\ with
\[r=(\log^*(n) +\Delta^{O(1/\eps)}) \cdot
    \left(\log \left(\frac{1}{\wmin(\eps)}\right)\right)^{O(1/\eps)}.\]
\end{theorem}
\medskip\noindent
The proof of Theorem~\ref{thm:distalg mwm} is based on the following lemma.
Recall that ignoring lightweight edges implies that $\frac{1}{\wmin(\eps)} \leq
\frac{n}{\eps}$.

\begin{lemma}\label{lemma:radius weight}
The probe radius of the \clocal-algorithm $\oracle_{L,T_k}$ is
  \[
    r_G(\oracle_{L,T_k}) \leq
(\log^*(n) +\Delta^{O(1/\eps)}) \cdot
    \left(\log \left(\frac{1}{\wmin(\eps)}\right)\right)^{O(1/\eps)}
  \]
\end{lemma}

\begin{proof}
  The description of the oracle $\oracle_{i,j}$ implies that the probe radius
  $r_G(\oracle_{i,j})$ satisfies the following recurrence:
  \begin{align*}
    r_G(\oracle_{i,j})&=
    \begin{cases}
      0 & \text{if $(i,j)=(1,0)$},\\
      \max\{r_G(\oracle_{\prev(i,j)}), r_G(\proca_{i,j})\}& \text{else}.
\end{cases}
\end{align*}
The description of the procedure $\proca_{i,j}$ implies
that the probe radius $r_G(\proca_{i,j})$ satisfies the
following recurrence:
  \begin{align*}
    r_{G}(\proca_{i,j})&\leq O(k) + \max\{
r_G(\oracle_{\prev(i,j)}) , r_G(\lmis(H_{i,j}))\}.
\end{align*}
The simulation of probes to $H_{i,j}$ implies an increase
in the radius by a factor of $2k+1$ in addition to the
radius of the probes. Hence,
\begin{align*}
  r_{G}(\lmis(H_{i,j})) &\leq (2k+1) \cdot   r_{H_{i,j}}(\lmis(H_{i,j})) \\
                        & ~~~~~~~~~~ + r_G(\textit{probe}(i,j,p)).
\end{align*}
The orientation of $H_{i,j}$ can be based on a coloring of the intersection graph
$G_k$. Hence, by Lemma~\ref{lemma:simul}
\begin{align*}
  r_{H_{i,j}}(\lmis(H_{i,j})) &\leq r_{G_k}(C(G_k)) + |C(G_k)|.
\end{align*}

By employing a distributed vertex coloring algorithm with $O(\log^*(n) + \poly(\Delta))$
rounds that uses $\poly(\Delta)$ colors, we obtain
\begin{align*}
  r_G(\oracle_{i,j}) &\leq O(k) + r_G(\lmis(H_{i,j}))\\
&\leq O(k)\cdot r_{H_{i,j}}(\lmis(H_{i,j})) + r_G(\textit{probe}(i,j,p))\\
&\leq r_G(\oracle_{\prev(i,j)}) \\
& ~~~~~~~~~~ + O\Big(k\cdot \log^* (n(G_k)) + \poly(\Delta(G_k))\Big),
\end{align*}
Since $\Delta(G_k) \leq (2k+1)^2 \Delta^{2k+1}$ and $n(G_k)
\leq n^{2k+1}$ we get that
  \begin{align*}
    & r_G(\oracle_{L,T_k})  \leq L\cdot T_k \cdot  O\left(k\cdot \log^*(n) +\poly\left(\Delta^{k}\right)\right) \\
    & ~ = O\left(\frac {1}{\eps} \cdot \log \left(\frac{1}{\eps}\right) \cdot W^{O\left(\frac{1}{\eps}\right)} \cdot
      \left(\frac{1}{\eps}\cdot \log^*(n) +\poly\left(\Delta^{\frac{1}{\eps}}\right)\right)\right)\\
    & ~ =\left(\log^*(n) +\poly\left(\Delta^{\frac{1}{\eps}}\right)\right) \cdot \poly\left(\frac{1}{\eps}\right) \cdot
    \poly\left(W^{\frac{1}{\eps}}\right),
  \end{align*}
and the lemma follows.
\end{proof}

\section{Future Work} In the full version we present an
improved algorithm for the $(1-\eps)$-approximate \mwm.
This improved algorithm computes a $(1-\eps)$-approximate
\mwm\ within $$O\left(\frac{1}{\eps^2} \cdot \log
\frac{1}{\eps} \right)\cdot \log^*n +
    \Delta^{O(1/\eps)}\cdot \log \left(\frac{1}{\wmin(\eps)}
         \right)$$
rounds.

\ifnum \icdcn=0
\bibliographystyle{alpha}
\else
\bibliographystyle{abbrv}
\fi
\bibliography{local-short}

\begin{thebibliography}{MRVX12}

\bibitem[ARVX12]{shaisoda2012}
N.~Alon, R.~Rubinfeld, S.~Vardi, and N.~Xie.
\newblock Space-efficient local computation algorithms.
\newblock In {\em SODA}, pages 1132--1139, 2012.

\bibitem[BE09]{barenboim2009distributed}
L.~Barenboim and M.~Elkin.
\newblock {Distributed ($\Delta$+ 1)-coloring in linear (in $\Delta$) time}.
\newblock In {\em STOC}, pages 111--120, 2009.

\bibitem[CHW08]{czygrinow2008fast}
A.~Czygrinow, M.~Ha{\'n}{\'c}kowiak, and W.~Wawrzyniak.
\newblock Fast distributed approximations in planar graphs.
\newblock In {\em DISC}, pages 78--92. Springer, 2008.

\bibitem[CV86]{cole1986deterministic}
R.~Cole and U.~Vishkin.
\newblock Deterministic coin tossing with applications to optimal parallel list
  ranking.
\newblock {\em Inf. and Cont.}, 70(1):32--53, 1986.

\bibitem[EMR14]{EMR14}
G.~Even, M.~Medina, and D.~Ron.
\newblock Deterministic stateless centralized local algorithms for bounded
  degree graphs.
\newblock {\em Accepted to ESA 2014}, 2014.

\bibitem[HK73]{hopcroft1973n}
J.~E. Hopcroft and R.~M. Karp.
\newblock An n\^{}5/2 algorithm for maximum matchings in bipartite graphs.
\newblock {\em SICOMP}, 2(4):225--231, 1973.

\bibitem[HV06]{HV06}
S.~Hougardy and D.~E Vinkemeier.
\newblock Approximating weighted matchings in parallel.
\newblock {\em IPL}, 99(3):119--123, 2006.

\bibitem[Kuh09]{kuhn2009weak}
F.~Kuhn.
\newblock Weak graph colorings: distributed algorithms and applications.
\newblock In {\em SPAA}, pages 138--144. ACM, 2009.

\bibitem[Lin92]{linial}
N.~Linial.
\newblock Locality in distributed graph algorithms.
\newblock {\em SICOMP}, 21(1):193--201, 1992.

\bibitem[LPSP08]{DBLP:conf/spaa/LotkerPP08}
Z.~Lotker, B.~Patt-Shamir, and S.~Pettie.
\newblock Improved distributed approximate matching.
\newblock In {\em SPAA}, pages 129--136, 2008.

\bibitem[LW08]{lenzen2008leveraging}
C.~Lenzen and R.~Wattenhofer.
\newblock Leveraging {L}inial's locality limit.
\newblock In {\em DISC}, pages 394--407. Springer, 2008.

\bibitem[MRVX12]{shaiicalp2012}
Y.~Mansour, A.~Rubinstein, S.~Vardi, and N.~Xie.
\newblock Converting online algorithms to local computation algorithms.
\newblock In {\em ICALP}, pages 653--664. 2012.

\bibitem[MV13]{shaiapprox2013}
Y.~Mansour and S.~Vardi.
\newblock A local computation approximation scheme to maximum matching.
\newblock In {\em APPROX}, pages 260--273, 2013.

\bibitem[NO08]{onak2008}
H.~N Nguyen and K.~Onak.
\newblock Constant-time approximation algorithms via local improvements.
\newblock In {\em FOCS}, pages 327--336, 2008.

\bibitem[Ona10]{onakthesis}
K.~Onak.
\newblock New sublinear methods in the struggle against classical problems,
  2010.

\bibitem[Pel00]{pelegbook}
D.~Peleg.
\newblock {\em Distributed computing: a locality-sensitive approach}, volume~5.
\newblock SIAM, 2000.

\bibitem[PR01]{panconesi2001some}
A.~Panconesi and R.~Rizzi.
\newblock Some simple distributed algorithms for sparse networks.
\newblock {\em Dist. Comp.}, 14(2):97--100, 2001.

\bibitem[PR07]{parnasron}
M.~Parnas and D.~Ron.
\newblock Approximating the minimum vertex cover in sublinear time and a
  connection to distributed algorithms.
\newblock {\em Theo. Comp. Sci.}, 381(1):183--196, 2007.

\bibitem[PS04]{DBLP:journals/ipl/PettieS04}
S.~Pettie and P.~Sanders.
\newblock A simpler linear time 2/3-epsilon approximation for maximum weight
  matching.
\newblock {\em IPL}, 91(6):271--276, 2004.

\bibitem[PS10]{panconesi2010fast}
A.~Panconesi and M.~Sozio.
\newblock Fast primal-dual distributed algorithms for scheduling and matching
  problems.
\newblock {\em Dist. Comp.}, 22(4):269--283, 2010.

\bibitem[RTVX11]{shaiics2011}
R.~Rubinfeld, G.~Tamir, S.~Vardi, and N.~Xie.
\newblock Fast local computation algorithms.
\newblock In {\em ICS}, pages 223--238, 2011.

\bibitem[Suo13]{dsurvey}
J.~Suomela.
\newblock Survey of local algorithms.
\newblock {\em ACM Comput. Surv.}, 45(2):24:1--24:40, 2013.

\bibitem[YYI12]{yoshida2009improved}
Y.~Yoshida, M.~Yamamoto, and H.~Ito.
\newblock Improved constant-time approximation algorithms for maximum matchings
  and other optimization problems.
\newblock {\em SICOMP}, 41(4):1074--1093, 2012.

\end{thebibliography}

\appendix
\section{Pseudo-Code for \clocall $(1-\eps)$-approximate
\mwmm-Algorithm}\label{sec:pseudo code}

\begin{algorithm}[H]\scriptsize
\caption{$\text{Global-APX-MWM}(G,\eps)$ -Onak's adaption~\cite{onakthesis} of the
  global algorithm of Pettie and Sanders~\cite{DBLP:journals/ipl/PettieS04}.} \label{alg:global
mwm}
\begin{algorithmic}[1]
\Require A graph $G=(V,E)$ with edge weights $w(e)\in
(0,1]$ that are integer powers of $1-\frac{\eps}{3}$ for
$0<\eps <1$. \Ensure  A $(1-\eps)$-approximate weighted
matching \State $k \gets\lceil \frac{3}{\eps}\rceil$.
\State $L \gets \Theta\left(\frac{1}{\eps} \cdot \log\left(
\frac{1}{\eps}\right)\right)$. \State $T_k = \Theta
\left(\frac{1}{\eps}\cdot \log
\left(\frac{1}{\wmin(\eps)}\right)\right)^{2k+1}$. \Comment
$T_k$ is an upper bound on the number of distinct gains of
augmenting paths of length at most $2k+1$. \State $M\gets
\emptyset$. \For {$i=1$ to $L$} \For {$j=1$ to $T_k$}
\State $P^*_{i,j}$ is a maximal set of vertex disjoint
$(M,[1,k])$-augmenting paths with gain $g_j$.
  \State $M\gets M \oplus E(P^*_{i,j})$.
  \EndFor  \EndFor
\State \textbf{Return} $M$.
\end{algorithmic}
\end{algorithm}
\begin{algorithm}[H]\scriptsize
\caption{$\text{Global-APX-MWM'}(G,\eps)$ -Rewriting of
Algorithm~\ref{alg:global mwm}
  using the intersection graph $H_{i,j}$.} \label{alg:global
mwm'}
\begin{algorithmic}[1]
\Require A graph $G=(V,E)$ with edge weights $w(e)\in
(0,1]$ that are integer powers of $1-\frac{\eps}{3}$ for
$0<\eps <1$. \Ensure  A $(1-\eps)$-approximate weighted
matching \State $k,L,T_k$ as in Algorithm~\ref{alg:global
mwm} \State $M_{1,0}\gets \emptyset$. \For {$i=1$ to $L$}
\For {$j=1$ to $T_k$} \State $P_{i,j} \gets $ set of
$(M_{\prev(i,j)},[1,k])$-augmenting paths with gain $g_j$.
  \State Construct the intersection graph $H_{i,j}$ over $P_{i,j}$.
  \State $P^*_{i,j} \gets \mis (H_{i,j})$.
  \State $M_{i,j} \gets M_{\prev(i,j)} \oplus E(P^*_{i,j})$.
  \EndFor  \EndFor
\State \textbf{Return} $M_{L,T_k}$.
\end{algorithmic}
\end{algorithm}

\begin{algorithm}[H]\scriptsize
\caption{$\oracle_{i,j}(e)$ - a recursive oracle for membership in the
  approximate weighted matching.} \label{alg:Oij}
\begin{algorithmic}[1]
\Require A query $e \in E$.
\Ensure Is $e$ an edge in the weighted matching $M_{i,j}$?
  \State If $(i,j)=(1,0)$ then return false.
  \State $\tau\gets \oracle_{\prev(i,j)}(e)$.
  \State $\rho\gets \proca_{i,j}(e)$.
  \State \textbf{Return} $\tau \oplus \rho$.
\end{algorithmic}
\end{algorithm}
%
\begin{algorithm}[H]\scriptsize
\caption{$\proca_{i,j}(e=(u,v))$ - a procedure for checking
membership of an edge $e$ in one
  of the paths in $P^*_{i,j}$.} \label{alg:Aij}
\begin{algorithmic}[1]
\Require An edge $e \in E$.
\Ensure Does $e$ belong to a path $p\in P^*_{i,j}$?
\State \textbf{Listing step:} \Comment Compute all shortest $M_{i-1}$-augmenting
paths that contain $e$.
  \State \quad $B_u \gets BFS_G(u)$ with depth $2k$.
  \State \quad $B_v \gets BFS_G(v)$ with depth $2k$.
  \State \quad For every edge $e'$ in the subgraph of $G$ induced by $B_u
\cup B_v$: $\chi_{e'}\gets \oracle_{\prev(i,j)}(e')$.
  \State \quad $P_{i,j}(e)\gets$ all $(M_{\prev(i,j)},[1,k])$-augmenting paths that
  contain $e$ with gain $g_j$
\State \textbf{\mis-step:} \Comment Check if one of the augmenting paths is in $P^*_i$.
  \State \quad For every $p\in P_{i,j}(e)$: If $\lmis(H_{i,j},p)$ \textbf{Return} true.
  \State \quad \textbf{Return} false.
\end{algorithmic}
\end{algorithm}
%
\begin{algorithm}[H]\scriptsize
  \caption{$\textit{probe}(i,j,p)$ - simulation of a probe to the intersection graph
    $H_{i,j}$ via probes to $G$.} \label{alg:prob Hij}
\begin{algorithmic}[1]
\Require A path $p \in P_i$ and the ability to probe $G$.
\Ensure The set of $(M_{\prev(i,j)}, k^*)$-augmenting paths with gain $g_j$ that intersect $p$.
\State For every $v\in p$ do
  \State \quad $B_v \gets BFS_G(v)$ with depth $2k+1$.
  \State \quad For every edge $e'\in B_v$: $\chi_e\gets \oracle_{\prev(i,j)}(e)$.
  \Comment determines whether the path is alternating and whether the endpoints
  are free.
  \State \quad $P_{i,j}(v)\gets$ all $(M_{\prev(i,j)},[1,k])$-augmenting paths that
  contain $v$ with gain $g_j$.
\item \textbf{Return} $\bigcup_{v\in p} P_{i,j}(v)$. 
\end{algorithmic}
\end{algorithm}

\end{document}